\documentclass[11pt]{article}
\usepackage[margin=1in]{geometry}
\usepackage{hyperref}
\usepackage{amsfonts,amsmath,amsthm,wrapfig}
\usepackage{comment,fullpage}
\usepackage[noend]{algorithmic}
\usepackage[boxed]{algorithm}
\usepackage{xspace}
\usepackage{framed}
\usepackage{xcolor}
\usepackage{verbatim}
\usepackage{tikz}
\usetikzlibrary{arrows}
\usepackage{enumitem}

\newtheorem{theorem}{Theorem}
\newtheorem{lemma}[theorem]{Lemma}

\newtheorem{claim}[theorem]{Claim}

\newtheorem{observation}[theorem]{Observation}

\newcommand{\polylog}{\ensuremath{\mathrm{polylog}}\xspace}

\newcounter{note}[section]

\date{\vspace{-5ex}}
\begin{document}

\title{Quasi-PTAS for Scheduling with Precedences using LP Hierarchies}

\author{
Shashwat Garg\thanks{Supported by the Netherlands Organisation for Scientific Research (NWO) under project no.~022.005.025.}\\
Eindhoven University of Technology \\
\href{mailto:s.garg@tue.nl}{s.garg@tue.nl}
}

\maketitle

\begin{abstract}
A central problem in scheduling is to schedule $n$ unit size jobs with \textit{precedence constraints} on $m$ identical machines so as to minimize the makespan. For $m=3$, it is not even known if the problem is NP-hard and this is one of the last open problems from the book of Garey and Johnson.

We show that for fixed $m$ and $\epsilon$, $\polylog(n)$ rounds of Sherali-Adams hierarchy applied to a natural LP of the problem provides a $(1+\epsilon)$-approximation algorithm running in quasi-polynomial time. This improves over the recent result of Levey and Rothvoss, who used $r=(\log n)^{O(\log \log n)}$ rounds of Sherali-Adams in order to get a $(1+\epsilon)$-approximation algorithm with a running time of $n^{O(r)}$.
\end{abstract}

\clearpage

\section{Introduction}
A central problem in scheduling is the following: suppose we are given $n$ unit jobs which have to be processed non-preemptively on $m$ identical machines. There is also a precedence order among the jobs: if $i\prec j$, then job $i$ has to be completed before $j$ can begin. The goal is to find a schedule of the jobs with the minimum makespan, which is defined as the time by which all the jobs have finished. 

This problem admits an easy $(2-\frac{1}{m})$ approximation algorithm which was given by Graham \cite{G66} in the 60's and is one of the landmark results in scheduling. This algorithm is known as the \textit{list-scheduling} algorithm and works as follows: at every time $t=1,2,\dots$, if there is an empty slot on any of the $m$ machines, schedule any \textit{available} job there, where a job is available if it is not yet scheduled and all the jobs which must precede it have already been scheduled. This simple greedy algorithm is essentially the best algorithm for the problem and for almost half a century it was an open problem whether one can get a better approximation algorithm. In fact, this was one of the ten open problems in Schuurman and Woeginger's influential list of open problems in scheduling \cite{Ger}. It was known since the 70's that it is NP-hard to get an approximation factor better than $4/3$ \cite{L78}. Slight improvements were given by Lam and Sethi \cite{L77} who gave a $2-\frac{2}{m}$ approximation algorithm, and Gengal and Ranade \cite{G08} who gave a $2-\frac{7}{3m+1}$ approximation algorithm for $m\ge 4$. Finally in 2010, Svensson \cite{Ola} showed that assuming a variant of Unique Games conjecture due to Bansal and Khot \cite{B09}, for any constant $\epsilon>0$ there is no $(2-\epsilon)$ approximation algorithm for the problem. 

However, this still leaves open the problem for the important case when $m$ is a constant. In fact in practice, usually the number of jobs are very large but there are only a few machines. Surprisingly, for $m=3$, it is not even known if the problem is NP-hard. This is one of the four problems from the book of Garey and Johnson \cite{GJ79} whose computational complexity is still unresolved. 

In order to get a better algorithm for the case when $m$ is a constant, a natural strategy is to write a linear program (LP), and for this problem, one such LP is the time-indexed LP \eqref{lp}, in which we first make a guess $T$ of the makespan and then solve the LP. The value of the LP is the smallest $T$ for which the LP is feasible, and the worst case ratio of the optimal makespan and the value of the LP is known as the integrality gap of the LP. It is well known that LP \eqref{lp} has an integrality gap of at least $2-\frac{2}{m+1}$ (see e.g. \cite{R16}), which suggests that one needs to look at stronger convex relaxations in order to get a better algorithm. Such a stronger convex relaxation can be obtained by applying a few rounds of a hierarchy to the LP, and in this paper, we will use the Sherali-Adams hierarchy \cite{Sher}. It is known that just one round of Sherali-Adams hierarchy reduces the integrality gap to $1$ for $m=2$ and thus, the problem can be solved exactly in this case (credited to Svensson in \cite{Rot13}). Claire Mathieu in \cite{Dag10} asked if one can get a $(1+\epsilon)$-approximation algorithm using $f(\epsilon, m)$ rounds of Sherali-Adams hierarchy for some function $f$ independent of $n$, which would imply a PTAS for the problem when $m$ is a constant. This is also Open Problem 1 in Bansal's recent list of open problems in scheduling \cite{Ban17}. 

To get some intuition behind why hierarchies should help in this problem, let us first look at the analysis for Graham's list-scheduling algorithm. At the end of this algorithm, the number of time slots which are busy, that is where all the $m$ machines have some job scheduled on them, is a lower bound on the optimum. Also, the number of non-busy time slots is a lower bound on the optimum. This is because there must be a \textit{chain} of jobs $j_1\prec j_2 \prec \dots \prec j_k$ such that one job from this chain is scheduled at each non-busy time, and the length of any chain in the instance is clearly a lower bound on the optimum. This implies that the makespan given by the algorithm, which is the sum of the number of busy and non-busy time slots, is a $2$-approximation of the optimum makespan, and a slightly more careful argument gives the guarantee of $(2-1/m)$. Now the key idea is that if the instance given to us has a maximum chain length of at most $\mathcal{\epsilon}$ times the optimal makespan, then Graham's list-scheduling algorithm already gives a $(1+\epsilon)$-approximation, and hierarchies provide, via conditionings, a good way to ``effectively" reduce the length of the chains in any given instance. 

Though the question of whether one can get a $(1+\epsilon)$-approximation algorithm using $f(\epsilon,m)$ rounds of Sherali-Adams hierarchy is still unresolved, a major breakthrough was made recently by Levey and Rothvoss \cite{R16}, who gave a $(1+\epsilon)$-approximation algorithm using $r=(\log n)^{O(m^2\log\log n/\epsilon^2)}$ rounds of Sherali-Adams. This gives an algorithm with a running time of $n^{O(r)}$, which is faster than exponential time but worse than quasi-polynomial time.

\subsection{Our Result}
In this paper, we improve over the result of Levey and Rothvoss \cite{R16} by giving a $(1+\epsilon)$-approximation algorithm which runs in quasi-polynomial time. Formally, we show the following:

\begin{theorem}
\label{thm:main}
The natural LP \eqref{lp} for the problem augmented with $r$ rounds of Sherali-Adams hierarchy has an integrality gap of at most $(1+\epsilon)$, where $r= O_{m,\epsilon}( \log^{O(m^2/\epsilon^2)}n)$. Moreover, there is a $(1+\epsilon)$-approximation algorithm for this problem running in time $n^{O(r)}$.
\end{theorem}

Throughout the paper, we use the notation $O_{m,\epsilon}(.)$ to hide factors depending only on $m$ and $\epsilon$. The natural LP for the problem is the following:
\begin{eqnarray}
\sum_{t=1}^T y_{jt} &=& 1 \qquad\qquad\quad \forall j\in[n] \nonumber\\
\sum_{j} y_{jt} &\le& m \qquad \qquad\quad\forall t\le T \label{lp}\\
\sum_{t'\le t}y_{jt'} &\ge&  \sum_{t'\le t+1}y_{it'}   \qquad \forall t\le T,\,\,\forall j\prec i   \nonumber\\  
y_{jt} &\ge& 0 \qquad \qquad \quad \forall t\le T,\,\,\forall j\in[n] \nonumber
\end{eqnarray}

\noindent
Here $T$ is our guess on the optimum makespan. In an integral solution, $y_{jt}=1$ if job $j$ is scheduled at time $t$, and $0$ otherwise. The first constraint ensures that each job is scheduled at exactly one time and the second constraint ensures that no more than $m$ jobs are scheduled at any time. The third constraint ensures that if $j\prec i$, then job $i$ can only be scheduled at a time strictly later than job $j$.

\subsection{Overview of Our Algorithm}

Let us first give an overview of the algorithm of Levey-Rothvoss \cite{R16} since our algorithm builds up on it.

\paragraph{Previous approach.} At a high-level, the algorithm of Levey-Rothvoss \cite{R16} works by constructing a laminar family of intervals, where the topmost level has one interval $[1,T]$ and each succeeding level is constructed by dividing each interval of the previous level into two equal sized intervals, as shown in the figure below. Thus, there are $(1+\log T)$ levels where level $\ell$ contains $2^{\ell}$ intervals, each of size $\frac{T}{2^{\ell}}$ for $\ell=0,1,\dots,\log T$. This laminar family can be thought of as being a binary tree of depth $\log T$ with the interval $[1,T]$ as the root and the level $\ell$ intervals being vertices at depth $\ell$. \\

\tikzstyle{decision}=[diamond, draw, fill=blue!50]
\tikzstyle{line}=[draw, -latex']
\tikzstyle{block}=[draw, rectangle, fill=gray!50, text width=16cm, minimum height=10mm, text centered]

\begin{tikzpicture}[scale=1]
\draw[fill=gray!50] (0,1) rectangle (12,0); \node at (-0.75,0.5) {Level $0$};
\draw[fill=gray!50] (0,0) rectangle (6,-1); \node at (-0.75,-0.5) {Level $1$};
\draw[fill=gray!50] (6,0) rectangle (12,-1);
\draw[fill=gray!50] (0,-1) rectangle (3,-2);\node at (-0.75,-1.5) {Level $2$};
\draw[fill=gray!50] (3,-1) rectangle (6,-2);
\draw[fill=gray!50] (6,-1) rectangle (9,-2);
\draw[fill=gray!50] (9,-1) rectangle (12,-2);
\draw[fill=black] (6,-2.5) circle (1mm);
\draw[fill=black] (6,-3) circle (1mm);
\draw[fill=black] (6,-3.5) circle (1mm);
\draw[fill=gray!50] (0,-4) rectangle (0.75,-5); \node at (0.375,-5.5) {$1$};
\draw[fill=gray!50] (0.75,-4) rectangle (1.5,-5); \node at (1.125,-5.5) {$2$};
\draw[fill=gray!50] (1.5,-4) rectangle (2.25,-5); \draw[fill=black] (2.15,-5.5) circle (0.3mm); \draw[fill=black] (2.45,-5.5) circle (0.3mm); \draw[fill=black] (2.75,-5.5) circle (0.3mm);
\draw[fill=gray!50] (2.25,-4) rectangle (3,-5); 
\draw[fill=gray!50] (3,-4) rectangle (3.75,-5);
\draw[fill=gray!50] (3.75,-4) rectangle (4.5,-5);
\draw[fill=gray!50] (4.5,-4) rectangle (5.25,-5);
\draw[fill=gray!50] (5.25,-4) rectangle (6,-5);
\draw[fill=gray!50] (6,-4) rectangle (6.75,-5);\node at (-1,-4.5) {Level $\log T$};
\draw[fill=gray!50] (6.75,-4) rectangle (7.5,-5);
\draw[fill=gray!50] (7.5,-4) rectangle (8.25,-5);
\draw[fill=gray!50] (8.25,-4) rectangle (9,-5);
\draw[fill=gray!50] (9,-4) rectangle (9.75,-5);
\draw[fill=gray!50] (9.75,-4) rectangle (10.5,-5);
\draw[fill=gray!50] (10.5,-4) rectangle (11.25,-5);
\draw[fill=gray!50] (11.25,-4) rectangle (12,-5); \node at (11.625,-5.5) {$T$};
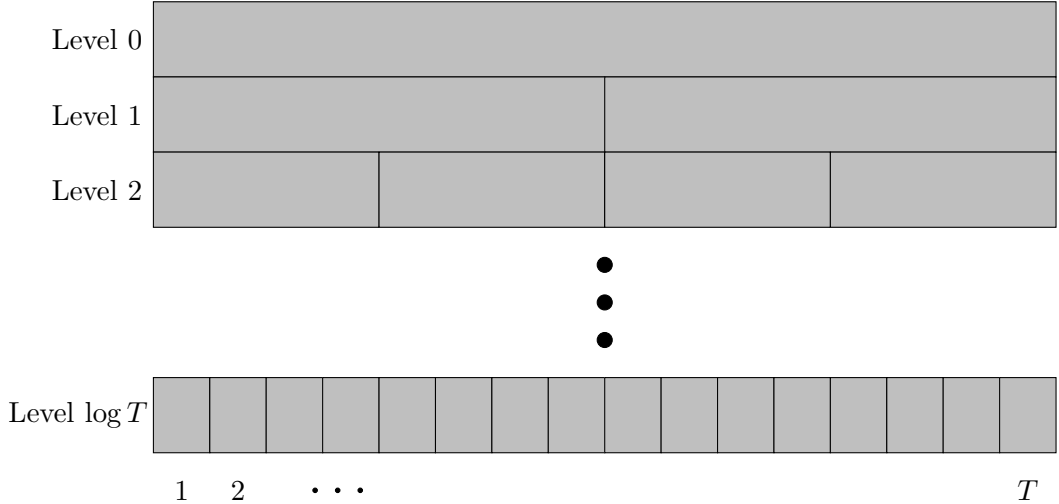
\node at (6,-6) {Figure $1:$ Construction of the laminar family used in the algorithm.};
\end{tikzpicture} 

Each job $j$ is first assigned to the smallest interval in this laminar family which fully contains the fractional support of $j$ as per the solution of the LP. Let $k=O(\log\log n)$ and let us call the top $k^2$ levels in the laminar family as the \textit{top levels} and the level succeeding it, that is the level $k^2$ as the \textit{bottom level}. Their algorithm conditions (see Section~\ref{sec:prelims} for the definition of conditioning) roughly $2^{k^2}$ times in order to reduce the maximum chain length among the jobs assigned to the top levels. Once the length of the chains in the top levels is reduced, the last $k$ of the top levels are discarded from the instance, and the sub-instances corresponding to each interval in the bottom level is recursively solved in order to get a partial schedule for all the jobs except those assigned to the top levels. The discarding of the $k$ levels is done in order to create a large gap between the top levels and the bottom level. Having such a gap makes it easier to schedule the remaining jobs in the top levels in the gaps of the partial schedule and Levey-Rothvoss \cite{R16} give an elegant algorithm to do this, provided that the maximum chain length among the jobs in the top levels is small. This step increases the makespan by at most a $(1+\epsilon/\log n)$ factor, which adds up to a loss of a $(1+\epsilon)$ factor in total over the at most $\log n$ depth of the recursion. Finally, one must also schedule the jobs in the $k$ levels which were discarded; to do this without increasing the makespan by more than a $(1+\epsilon)$ factor, it suffices to ensure that these $k$ discarded levels contain at most an $\epsilon$ fraction of the jobs contained in the top levels. Let us call such a set of $k$ consecutive levels, which contains at most an $\epsilon$ fraction of the number of jobs in the levels above it, as a \textit{good batch}.

Now the reason they had to condition $2^{k^2}=2^{O((\log\log n)^2)}$ times, which leads to the running time of $n^{O(2^{k^2})}$, comes from the fact that they condition on every interval in the top $k^2$ levels. And this is necessary to ensure that a good batch exists. For example, the number of jobs contained in the levels $[pk,(p+1)k)$ may be about $e^{\epsilon p}\cdot\left(\epsilon T/\log n\right)$ for all $p < (1/\epsilon)\ln(m\log n) \approx k$, in which case there is no good batch in the first $o(k^2)$ levels.

\paragraph{Our approach.} To get around the above issue, we observe the following: if, after conditioning on only the top $Ck$ levels, where $C=O(1/\epsilon^2)$ is a big enough constant, there does not exist a good batch in the top $Ck$ levels, then in fact a $(1-\epsilon)$ fraction of the jobs in the top $Ck$ levels must lie in the last $(1/\epsilon^2)k$ levels, that is, in levels from $Ck-(1/\epsilon^2)k$ to $Ck-1$. This implies that we can discard the jobs in the first $Ck-(1/\epsilon^2)k$ levels by charging them to the jobs in the levels from $Ck-(1/\epsilon^2)k$ to $Ck-1$, and in doing so we only discard an $\epsilon$ fraction of the total number of jobs.

Notice that we have only conditioned about $2^{Ck}=\polylog(n)$ times till now as there are these many intervals in the top $Ck$ levels. The next crucial observation is that after deleting the top $Ck-(1/\epsilon^2)k$ levels, the sub-instances defined by each of the subtrees rooted at the intervals on the level $Ck-(1/\epsilon^2)k$ can be solved independently of each other. This means that we can perform conditioning in parallel on each such sub-instance, and thus in total, we will condition at most $2^{Ck}\cdot\log T=\polylog(n)$ times, since the depth of the recursion is at most the height of the tree.

Now it might also happen that we already find a good batch in the top $Ck$ levels and in this case, we follow a strategy similar to Levey-Rothvoss \cite{R16} by recursing on the bottom intervals to find a partial schedule and fitting the jobs in the top levels in this partial schedule. This step might discard an $\epsilon$ fraction of the jobs in the top levels. These two cases, one where we recurse because there is no good batch in the top $Ck$ levels and one where we recurse because there is a good batch in the top $Ck$ levels, might interleave in a complicated manner. We show that the number of jobs ever discarded in the algorithm due to each type of recursion is at most an $O(\epsilon)$ fraction of the total number of jobs, which implies that we can achieve a makespan of $(1+\epsilon)T$. 

The above high-level description skims over a few important issues. One big challenge in the above approach is to ensure that the number of jobs discarded in the cases where we do not find a good batch stays small during the whole algorithm. Even though this is the case in one such recursive call, this might not happen over all the recursive calls taken together and we might end up discarding a constant fraction of the jobs. To get over this obstacle, we carefully control which interval each job is assigned to: if a job $j$ is assigned to an interval $I$ but after conditioning on some job $i\neq j$ which is assigned to a level lower than $j$, the fractional support of $j$ shrinks to a sub-interval of $I$, then we will still keep $j$ assigned to $I$, rather than moving it down the laminar family. This ensures that each job is not charged more than once for discarded jobs and thus the total number of discarded jobs is at most $\epsilon n$. This however slightly changes the way jobs are assigned to intervals and the techniques developed by Levey and Rothvoss \cite{R16} cannot be immediately applied to fit the jobs of the top levels in the partial schedule of the bottom levels in the case when a good batch exists. 
To tackle this issue, we will allow each job in the top levels to be scheduled outside of its (current) fractional support as long as it doesn't violate the precedence constraints with the jobs in the bottom levels. With this modification, we will be able to fit the jobs in the top levels in the partial schedule of the bottom levels without discarding more than an $\epsilon$ fraction of the top jobs. This implies that in both types of recursions, we only discard an $O(\epsilon)$ fraction of the jobs.

\section{Preliminaries on Sherali-Adams Hierarchy} 
\label{sec:prelims} 

In this section, we state the basic facts about Sherali-Adams hierarchy which we will need. We refer the reader to the excellent surveys \cite{Lau, Tul, Rot13} for a more extensive introduction to hierarchies.

Consider a linear program with $n$ variables $y_1,\dots,y_n$ where for each $i\in [n]$, $0\le y_i\le 1$. For $s\ge 0$, the $s^{th}$-round Sherali-Adams lift of this linear program is another linear program with variables $y^{(s)}_{S}$ for each $S\subseteq [n]$ satisfying $|S|\le s+1$, and some additional constraints. We will often denote $y^{(s)}_{\{i\}}$ by $y^{(s)}_i$ for simplicity. 

If we think of $y_i$ as the probability that $y_i=1$, intuitively the variables $y^{(s)}_S$ should equal the probability that each $i\in S$ has $y_i=1$, that is we would like to have that $y^{(s)}_S=\Pi_{i\in S} y_i$. As these constraints are not convex, we can only impose some linear conditions implied by them. In particular, for every constraint $a^Ty\le b$ of the starting LP, we add, for every $S,T\subseteq [n]$ such that $|S|+|T|\le s$, a new constraint given by 
\begin{equation}\label{eqn:sa}  \sum_{T'\subseteq T} (-1)^{|T'|}  \left(\sum_{i=1}^n a_i y^{(s)}_{S\cup T' \cup \{i\} }  -b y^{(s)}_{S\cup T'} \right)  \le 0 . \end{equation}

\noindent
If $y^{(s)}_S=\Pi_{i\in S}y_i$ was indeed true for all $|S|\le s+1$, then the above inequality can be succinctly written as $(a^Ty-b)\cdot\Pi_{i\in S}y_i \cdot\Pi_{i\in T}(1-y_i)  \le 0$, and are thus valid constraints for all $0-1$ solutions.

Observe that an $s^{th}$-round Sherali-Adams lift of an LP with $n$ variables and $m$ constraints is just another LP with $n^{O(s)}$ variables and $m\cdot n^{O(s)}$ constraints. Letting $y^{(s)}$ denote a feasible solution of the $s^{th}$-round Sherali-Adams lift, $y^{(s)}$ is also feasible for all $s'\le s$ rounds of Sherali-Adams and in particular is a feasible solution of the starting LP.

\paragraph{Conditioning.} Given a feasible solution $y^{(s)}$ of the $s^{th}$-round Sherali-Adams lift and $i\in[n]$ such that $y^{(s)}_i>0$, then we can \textit{condition on the event $y_i=1$} to get a feasible solution $z^{(s-1)}$ of the $(s-1)^{th}$-round Sherali-Adams lift defined as
\[ z^{(s-1)}_S = \frac{y^{(s)} _{S\cup \{i\}}  }{y^{(s)}_i}  \qquad\,\, \forall S: |S|\le s .\]

\noindent
The fact that $z^{(s-1)}$ is a feasible solution of the $(s-1)^{th}$-round Sherali-Adams lift follows easily from \eqref{eqn:sa}. Moreover, $z$ satisfies $z^{(s-1)}_i=1$ and the following useful property:

\begin{observation}
\label{obs:condsupp}
If for some $j\in[n]$, $y^{(s)}_j=0$ and we condition on $y_i=1$ for any $i\in [n]$, then $z^{(s-1)}_j=0$.
\end{observation}
\begin{proof}
Using the Sherali-Adams lift \eqref{eqn:sa} of the constraint $y_i\le 1$ with $S=\{j\}$ and $T=\phi$, we get
\[  y^{(s)}_{\{i, j\}} \le y^{(s)}_{j}.  \]
This gives
\[ z^{(s-1)}_j =\frac{y^{(s)} _{\{i,j\}}  }{y^{(s)}_i} \le \frac{y^{(s)} _{j}  }{y^{(s)}_i} =0. \] 
\end{proof}

One can think of the solution $z^{(s-1)}_S$ as giving the conditional probability of $y^{(s)}_S=1$ given $y^{(s)}_i=1$. By conditioning on a variable $y_i$ to be $1$, we will mean that we replace the current fractional solution $y^{(s)}$ with the fractional solution $z^{(s-1)}$ as in above. Observation~\ref{obs:condsupp} implies that conditioning can never increase the support of any variable, or in other words, if the probability that $y_j=1$ is zero, then the conditional probability that $y_j=1$ conditioned on $y_i=1$, is also zero.

%
%
%
%
\section{Algorithm}
Before we describe our algorithm, we first develop some notation. Let $T$ denote the value of the LP \eqref{lp} and let $y^{(s)}$ denote the feasible solution of the $s^{th}$-round Sherali-Adams lift of the LP we get after we condition $r-s$ times in the algorithm. We will say that we are in \textit{round $s$} of the algorithm if we have conditioned $r-s$ times so far. So we will start the algorithm in round $r$ with solution $y^{(r)}$, and if we condition in round $s$, we go to round $s-1$ with solution  $y^{(s-1)}$. 

For each job $j$, define the \textit{fractional support interval of $j$ in round $s$} as $F^{(s)}_j:=[ r^s_j, d^s_j   ]$, where $r^s_j$ is the smallest time $t$ for which $y^{(s)}_{jt}>0$ and $d^s_j$ is the largest time for which $y^{(s)}_{jt}>0$ ($r_j$ and $d_j$ are used to symbolize release time and deadline). In other words, $F^{(s)}_j$ is the minimal interval which fully contains the fractional support of job $j$ in $y^{(s)}$. By Observation~\ref{obs:condsupp}, upon conditioning, the fractional support interval can only shrink, that is $F^{(s-1)}_j \subseteq  F^{(s)}_j$.

For each job $j$, we also define a \textit{support interval} $S^{(s)}_j$. We will initially set $S^{(r)}_j:=F^{(r)}_j$. In later rounds, we will update $S^{(s)}_j$ in such a way that $F^{(s)}_j \subseteq S^{(s)}_j \subseteq F^{(r)}_j$. Intuitively, $S^{(s)}_j$ reflects our knowledge in round $s$ of where $j$ can be scheduled. Notice that we might schedule $j$ outside of the fractional support interval $F^{(s)}_j$.

A schedule of jobs is called a \textit{feasible schedule} if it satisfies the precedence constraints among all the jobs and a \textit{partial feasible schedule} if it schedules some of the jobs and satisfies the precedence constraints among them. In order to get a feasible schedule of all the jobs with a makespan of at most $(1+\epsilon)T$, it suffices to show the following:

\begin{theorem}
\label{thm:main2}
We can find a partial feasible schedule $\sigma:[n] \rightarrow[T]\cup \{ \text{DISCARDED} \}$ such that $\sigma(j)=\text{DISCARDED}$ for at most $\epsilon T$ jobs.
\end{theorem}

Clearly a schedule $\sigma$ as in Theorem~\ref{thm:main2} has makespan at most $T$. Having such a partial feasible schedule, we can easily convert it to a feasible schedule of all the $n$ jobs with a makespan of at most $(1+\epsilon )T$: iterate through every job $j$ discarded in $\sigma$ and find the earliest time $t$ by when all the jobs which must precede $j$ have either already been scheduled or are as of yet discarded. Create a new time slot between times $t$ and $t+1$ containing only job $j$. This increases the makespan by one for every job discarded in $\sigma$.

\paragraph{Laminar Family.} A laminar family of intervals is defined in the following manner. The topmost level, level $0$, has one interval $[1,T]$. Each succeeding level is constructed by dividing each interval of the previous level into two equal sized intervals\footnote{Without loss of generality, $T$ is a power of $2$. Otherwise, we can add a few dummy jobs at the end which must succeed all other jobs and which make $T$ a power of $2$.}. Thus there are $(1+\log T)$ levels, where level $\ell$ contains $2^{\ell}$ intervals each of size $\frac{T}{2^{\ell}}$ for $\ell=0,1,\dots,\log T$. This laminar family can be thought of as being a binary tree of depth $\log T$ with the interval $[1,T]$ as the root, and the level $\ell$ intervals as being vertices at depth $\ell$. 

Let $\mathcal{I}_{\ell}$ denote the set of intervals at level $\ell$ of the laminar family. For an interval $I\in\mathcal{I}_{\ell}$, a \textit{sub-interval of $I$} is any interval $I'\subseteq I$ of the laminar family, including $I$ itself; and $I_{\text{left}},I_{\text{right}}$ will denote the left and right sub-intervals respectively of $I$ in $\mathcal{I}_{\ell+1}$. By the midpoint of $I$, we will mean the right boundary of $I_{\text{left}}$.

Job $j$ is \textit{assigned} to interval $I$ in round $s$ if $I$ is the smallest interval in the laminar family such that $S^{(s)}_j\subseteq I$. This assignment of jobs to intervals depends on $S^{(s)}_j$ and will change as $s$ and $S^{(s)}_j$ change during the algorithm. Let $I^{(s)}(j)$ denote the interval to which $j$ is assigned in round $s$ of the algorithm. For an interval $I$ in the laminar family, let $\mathcal{J}^{(s)}(I)$ denote the set of jobs assigned to $I$, and let $\mathcal{J}^{(s)}(\mathcal{I}_{\ell})$ denote the set of jobs assigned to intervals in $\mathcal{I}_{\ell}$ in round $s$ of the algorithm. 

\paragraph{Batches.} Let $k=\log(\frac{32m}{\epsilon} \cdot \log n)$. For $p\ge 0$, define the $p^{th}$ batch as
\[ \mathcal{B}_p =\{  \mathcal{I}_{pk},\mathcal{I}_{pk+1},\dots,\mathcal{I}_{(p+1)k-1}\} . \]
That is, it denotes the set of $k$ consecutive levels starting from level $pk$ till level $(p+1)k-1$. Let $\mathcal{J}^{(s)}(\mathcal{B}_{p})$ denote the set of jobs assigned to intervals in batch $p$ in round $s$. 
 Batch $\mathcal{B}_p$ for $p\ge 1$ is called a \textit{good batch with respect to $[T]$ in round $s$} if 
\begin{equation} \label{eqn:gb} |\mathcal{J}^{(s)}(\mathcal{B}_{p})| \le \frac{\epsilon}{4m} \sum_{i=0}^{p-1} |\mathcal{J}^{(s)}(\mathcal{B}_{i})|. \end{equation}
We will omit the ``with respect to $[T]$" if it is clear from the context that we start the summation in the right hand side of \eqref{eqn:gb} from the first batch in $[T]$. Similarly, we will omit the ``in round $s$" if $s$ is clear from the context.

\paragraph{Algorithm.} We can now describe our algorithm and split its description in two steps for clearer exposition. Let $C=2(4m/\epsilon)^2+1,k=\log(\frac{32m}{\epsilon} \cdot \log n)$ and $\delta=\frac{\epsilon}{8mCk2^{Ck}\log n }$. The reader can think of these parameters as being $k=\Theta_{m,\epsilon}(\log\log n)$ and $\delta=\Theta_{m,\epsilon}(1/\polylog(n))$. $s$ will always denote the current round of the algorithm, unless otherwise specified. We initialise $s:=r$ and for each job $j$, $S^{(r)}_j:=F^{(r)}_j$.\\


\textbf{Schedule($y^{(r)}, T$): \footnote{When the algorithm is called on an interval of length $L\le 2^{Ck}$, we can just ``brute force" by conditioning $mL$ times to find an exact solution. We avoid writing this explicitly in the algorithm for simplicity.}}
\begin{enumerate}
\item \textbf{Step 1: Reducing chain length in the top $qk\le Ck$ levels} 

In this step, we will reduce the length of the chains in each interval $I$ in the top $qk$ levels of the laminar family to at most $\delta |I |$, for some $q\le C$. This is done by going down the levels, starting from level $0$ till level $qk-1$, where $q$ is chosen such that
\begin{enumerate}
\item after having conditioned on all the levels from $0$ to $qk-1$, $\mathcal{B}_{q-1}$ is a good batch, or 
\item we have already conditioned on the top $Ck$ levels and found no good batch, in which case we set $q=C$. 
\end{enumerate}
The conditioning on the levels and update of $S_j$'s is done as follows. For $\ell=0,1,\dots,qk-1$:
\begin{itemize}
\item Let $s_{\text{old}}=s$ and for each $j$, let $\ell(j)$ denote the level of the interval $I^{(s_{\text{old}})}(j)$, that is the level to which $j$ is assigned at the beginning of this iteration of the loop.
\item We go over every interval $I\in\mathcal{I}_{\ell}$ and do the following: if $\mathcal{J}^{(s)}(I)$ has a chain of length more than $\delta |I|$, let $j$ be the first job in this chain. We condition on $j$ lying in $I_{\text{right}}$. 

After every conditioning, update $s:=s-1$ and set $S^{(s)}_j$ for every job $j$ as follows: 
\begin{itemize}
\item if $\ell(j)<\ell$, let $m_j$ denote the midpoint of $I^{(s_{\text{old}})}(j)$ and $[t_r,t_d]:=F_j^{(s)}$. If $F_j^{(s)} \subseteq I_{\text{left}}$, then we set $S^{(s)}_j:=[t_r,m_j+1]$, and if $F_j^{(s)} \subseteq I_{\text{right}}$, then we set $S^{(s)}_j:=[m_j,t_d]$. Otherwise, set $S^{(s)}_j:=F_j^{(s)}$.
\item if $\ell(j)\ge\ell$, set $S^{(s)}_j:=F_j^{(s)}$. 
\end{itemize}
\end{itemize}
That is, the support intervals $S^{(s)}_j$ are set such that if we condition on jobs in level $\ell$, then the jobs assigned to a level $\ell'<\ell$ before the conditionings stay assigned to level $\ell'$, and for all other jobs, $S^{(s)}_j$ equals the fractional support interval $F_j^{(s)}$. 

\item \textbf{Step 2: Recursion}

There are two cases to consider here, depending on which of $(a)$ or $(b)$ took place in the previous step. 
\begin{enumerate}[label=(\roman*)]
\item \textbf{If $(a)$ occured, perform a recursion of type $1$. } 

This step is similar to the algorithm of \cite{R16}. We discard all the jobs in the good batch $\mathcal{B}_{q-1}$. Then for each interval $I\in\mathcal{I}_{qk}$, we recursively call Schedule($y^{(s)}, I$) to obtain a schedule $\tilde{\sigma}_I$, which are put together to form a partial feasible schedule $\tilde{\sigma}$ for the jobs assigned to a level $\ell\ge qk$.

Then we fit the jobs in the top levels, that is the jobs in $\mathcal{J}^{(s)}(\mathcal{B}_0) \cup \dots \cup\mathcal{J}^{(s)}(\mathcal{B}_{q-2})$ in the empty slots in $\tilde{\sigma}$. We give more details of how this is done in Section~\ref{sec:disc1}. Some jobs in the top levels will be discarded while doing this. 

Call this step a \textit{recursion of type $1$}. The number of jobs discarded in this step, that is the jobs in batch $\mathcal{B}_{q-1}$ along with the jobs in the top levels which are discarded, will be referred to as the jobs discarded due to this step. Notice that this does not include the jobs discarded in each recursive call to the intervals in $\mathcal{I}_{qk}$.

\item \textbf{If $(b)$ occured, perform a recursion of type $2$. } 

In this case, we discard all the jobs in $\mathcal{J}^{(s)}(\mathcal{B}_0) \cup \dots \cup \mathcal{J}^{(s)}(\mathcal{B}_{C-(4m/\epsilon)^2-1})$. Then for each interval $I\in \mathcal{I}_{(C-(4m/\epsilon)^2)k}$, we recursively call Schedule($y^{(s)}, I$) to get a schedule $\sigma_I$ which are put together to form a partial feasible schedule $\sigma$ for all the jobs assigned to a level $\ell\ge (C-(4m/\epsilon)^2)k$.

Call this step a \textit{recursion of type 2}. The number of jobs discarded in this step, that is the jobs in batches $\mathcal{B}_0, \dots, \mathcal{B}_{C-(4m/\epsilon)^2-1}$ will be referred to as the jobs discarded due to this step. Just like before, this does not include the jobs discarded in each recursive call to the intervals in $\mathcal{I}_{(C-(4m/\epsilon)^2)k}$.

\end{enumerate}

In each type of recursion, we recurse on multiple sub-instances defined by intervals of some level. It is important that the recursions on these sub-instances are done independently of each other. That is, we pass the same (current) Sherali-Adams solution to each recursive call, and conditionings done in one recursive call are independent of conditionings done in any other recursive call, and thus do not affect the fractional solution of any other recursive call.  
\end{enumerate}

\section{Analysis}
In this section, we prove Theorem~\ref{thm:main2} which will imply Theorem~\ref{thm:main}. We split the analysis into two parts: in the first part, we give a bound on the number of rounds of Sherali-Adams needed in the algorithm. In the second part, we show that we discard at most $\epsilon T$ jobs during the algorithm and schedule all other jobs by time $T$, thus proving Theorem~\ref{thm:main2}.

But first, we need to show that the algorithm is well-defined.

\begin{observation}
\label{obs:welldef}
In step $1$, when we condition on an interval $I$ by finding a chain $\mathcal{C}$ in $I$ and conditioning the first job $j$ in this chain to lie in $I_{\text{right}}$, this is possible to do. Moreover, this assigns every job in $\mathcal{C}$ to a sub-interval of $I_{\text{right}}$.
\end{observation}
\begin{proof}
For the first part of the observation, we need to show that $F^{(s)}_j\cap I_{\text{right}} \neq \phi$, where $s$ is the round of the algorithm just before we condition on $j$ in $I$. As $j$ is assigned to $I$ in round $s$, it must be that $S_j^{(s)}\cap I_{\text{right}} \neq \phi$. The support intervals are updated in a way such that we can only have $S^{(s)}_j\neq F^{(s)}_j$ after we condition on a level below that of $j$. But because we always condition on the levels from top to bottom, we must have $S^{(s)}_j=F^{(s)}_j$. This proves the first part of the observation.
%

The moreover part follows easily now since every other job $i\in \mathcal{C}$ satisfies $j\prec i$ and must start scheduling only after $j$.
\end{proof}

\subsection{Bounding number of rounds of Sherali-Adams}
Let $r(|I|)$ denote the number of rounds of Sherali-Adams the algorithm uses when run on the instance defined by the subtree rooted at interval $I$ of the laminar family. Our goal in this subsection is to show $r(T)\le r$ for $r=O_{m,\epsilon}(\log^{O(m^2/\epsilon^2)} n) $. 

We first give an upper bound on the number of conditionings done in one interval $I$.

\begin{lemma}
\label{lem:condinter2}
The algorithm conditions at most $m/\delta$ times on any interval $I$ in step $1$.
\end{lemma}
\begin{proof}
Let $s_{\text{old}}$ denote the round of the algorithm just before we start to condition in $I$, and let $\ell\ge 0$ be such that $I\in \mathcal{I}_{\ell}$. Each time we condition in $I$, we assign at least $\delta |I|$ jobs in $I$ to a sub-interval of $I_{\text{right}}$ (by Observation~\ref{obs:welldef}).

Also, no job assigned to a level $\ell'<\ell$ in round $s_{\text{old}}$ moves down the laminar family during conditionings done in $I$. And for all other jobs, they only get assigned to a sub-interval. Thus no new job is assigned to $I$ while we are conditioning in $I$. 

Using $F_j^{(s_{\text{old}})}\subseteq S_j^{(s_{\text{old}})}$ and the second constraint of LP \eqref{lp}, there can be at most $m|I|$ jobs in total assigned to $I$ in round $s_{\text{old}}$. Thus, the number of times we condition in $I$ is at most
\[ \frac{m|I|}{\delta |I|} =\frac{m}{\delta}. \] 
\end{proof}

\begin{lemma}
\label{lem:condinter}
The algorithm conditions at most $2^{Ck}m/\delta$ times in step $1$ of the algorithm.
\end{lemma}
\begin{proof}
By Lemma~\ref{lem:condinter2}, we condition at most $m/\delta$ times per interval. As we condition on the topmost $qk$ levels and hence on at most $2^{qk}\le 2^{Ck}$ intervals, we condition at most $2^{Ck}m/\delta$ times in step $1$.
\end{proof}

In step 2 of the algorithm, if we do a recursion of type 1 then we recurse on every interval at level $qk\ge k$. Otherwise, if we do a recursion of type 2 then we recurse on every interval at level $(C-(4m/\epsilon)^2)k\ge k$. In either case we recurse on every interval of some level $\ell \ge k$ and thus on an interval of size at most $T/2^k$. Because the conditionings done in one recursive call are done independently of the conditionings in any other recursive call, the total number of rounds of Sherali-Adams we need can be bounded by the following recurrence:
\[ r(T) \le \frac{2^{Ck}m}{\delta}+r(T/2^k)\]
where the base case is $r(2^{Ck})= 2^{Ck}m$, and thus we get
\begin{eqnarray*}
r(T) &\le&\frac{2^{Ck}m\log T}{\delta} \\
& \le & \frac{8m^2Ck (\log^2 n) 2^{2Ck}}{\epsilon} \\
&  =  & O_{m,\epsilon}(k(\log n)^{4+\frac{64m^2}{\epsilon^2}}) \\
&= & O_{m,\epsilon}((\log n)^{5+\frac{64m^2}{\epsilon^2}})=r.
\end{eqnarray*}

\subsection{Bounding number of jobs discarded}

In this subsection, we bound the number of jobs discarded in the algorithm and show that it is at most $\epsilon T$. We will separately bound the number of jobs discarded due to recursions of type $1$ and recursions of type $2$ and show that each is at most $\epsilon T/2$. The former uses a result proved by Levey and Rothvoss \cite{R16} but which needs to be heavily adapted to our algorithm. The latter uses a simple charging argument. 

\subsubsection{Jobs discarded due to recursions of type 1.}
\label{sec:disc1}
Suppose we perform a recursion of type $1$ when the algorithm is called on the interval $I$ of the laminar family. To be consistent with the notation of \cite{R16}, we will call the set of jobs $\mathcal{J}^{(s)}(\mathcal{B}_{q-1})$ as $\mathcal{J}_{\text{middle}}$, the set of jobs $\mathcal{J}^{(s)}(\mathcal{B}_0) \cup \dots \cup\mathcal{J}^{(s)}(\mathcal{B}_{q-2})$ as $\mathcal{J}_{\text{top}}$ and the jobs in the levels below these as $\mathcal{J}_{\text{bottom}}$ (here we are reindexing the batches such that the first level starts from interval $I$).

\begin{claim}
\label{claim:mid}
\[ |\mathcal{J}_{\text{middle}}| \le \frac{\epsilon}{4m} |\mathcal{J}_{\text{top}}|.\]
\end{claim}
\begin{proof}
Follows from the fact that $\mathcal{B}_{q-1}$ is a good batch and \eqref{eqn:gb}.
\end{proof}

After discarding all the jobs in $\mathcal{J}_{\text{middle}}$, the algorithm recursively finds a partial feasible schedule $\tilde{\sigma}$ of the jobs in $\mathcal{J}_{\text{bottom}}$. Let $\mathcal{J}' \subseteq \mathcal{J}_{\text{bottom}}$ be the set of jobs scheduled by $\tilde{\sigma}$. The algorithm will then attempt to extend $\tilde{\sigma}$ to a schedule $\sigma$ of the jobs in $\mathcal{J}_{\text{top}} \cup \mathcal{J}'$. We will be able to do this by discarding only a few jobs from $\mathcal{J}_{\text{top}}$. More formally:

\begin{lemma}
\label{lem:discardrec1}
When the algorithm is called on an interval $I$, we can extend $\tilde{\sigma}$ to a feasible schedule $\sigma$ of the jobs in $(\mathcal{J}_{\text{top}} \setminus  \mathcal{J}_{\text{discard}}) \cup \mathcal{J}'$ where 
\[|\mathcal{J}_{\text{discard}}| \le \frac{\epsilon |I|}{4\log n}.\]
\end{lemma}

Before going to the proof of Lemma~\ref{lem:discardrec1}, let us first see how it implies that we discard at most $\epsilon T/2$ jobs in all recursions of type $1$.

\begin{lemma}
Total number of jobs discarded in all recursions of type $1$ during the algorithm is at most $\epsilon T/2$.
\end{lemma}
\begin{proof}
Using Claim~\ref{claim:mid} and Lemma~\ref{lem:discardrec1}, if we perform a recursion of type $1$ when the algorithm is called on the interval $I$, the number of jobs discarded is at most
\[  \frac{\epsilon}{4m} |\mathcal{J}_{\text{top}}|+ \frac{\epsilon |I|}{4\log n} .\]
Over all recursions of type $1$, the first term sums up to at most $\epsilon n/4m \le \epsilon T/4$. For any $\ell\ge 0$, the second term sums up to $\frac{\epsilon T}{4\log n}$ over all intervals $I\in \mathcal{I}_{\ell}$. As there are at most $\log n$ levels, the second term also sums up to $\epsilon T/4$ over all recursions of type $1$. 
\end{proof}

We now come to the proof of Lemma~\ref{lem:discardrec1}. Without loss of generality and for easier notation, we will take $I=[T]$. A similar result was proved in \cite{R16} but we need to adapt their result to our setting before we can use it. Let us first mention what they proved. We need a bit of notation before that.

Let the intervals in $\mathcal{I}_{qk}$ be $I_1, \dots, I_{2^{qk}}$ where $I_p=[ (p-1)\frac{T}{2^{qk}} +1 ,p\frac{T}{2^{qk}} ]$. For any time interval $A=I_a\cup I_{a+1}\cup \dots\cup I_b$ where $a\le b$, define
\[ext(A) = I_{\max\{a-1,1\}}\cup \dots\cup I_{\min\{b+1,2^{qk}\}} .\]
In other words, we just extend $A$ by one interval from $\mathcal{I}_{qk}$ at either end if possible. 

For a job $j\in \mathcal{J}_{\text{top}}$, denote the interval it is assigned to by $I(j)$ ($s$ is implicitly fixed as we do not condition in this step) and let $m_j$ denote the midpoint of $I(j)$.

The following theorem is proved in $\cite{R16}$ though not stated in this form. For this reason, we show its proof in Appendix.
\begin{theorem}\cite{R16}
\label{thm:borrowed}
Suppose we are given a feasible schedule $\tilde{\sigma}$ of the jobs in $\mathcal{J}'\subseteq \mathcal{J}_{\text{bottom}}$ and let the maximum chain length among jobs in $\mathcal{J}_{\text{top}}$ be $\mathcal{C}$. Suppose we are also given for each $j\in \mathcal{J}_{\text{top}}$, an interval $[r_j,d_j]$ such that: 
\begin{enumerate}
\item $[r_j,d_j]=I_a\cup I_{a+1}\cup \dots\cup I_b$ for some $a\le b$.
\item if $j\prec i$ for some $i\in \mathcal{J}_{\text{top}}$, then $r_j\le r_i$ and $d_j\le d_i$.
\item if $j\prec i$ or $i\prec j$ for some $i\in \mathcal{J}'$, then $\tilde{\sigma}(i)\not\in[r_j,d_j]$.
\item $\mathcal{F}^{(s)}_j\subseteq ext([r_j,d_j])$ and $m_j$ lies in the interior of $ext([r_j,d_j])$.
\end{enumerate}
Then, we can extend $\tilde{\sigma}$ to a feasible schedule $\sigma$ of $(\mathcal{J}_{\text{top}} \setminus \mathcal{J}_{\text{discard}}  ) \cup\mathcal{J}'$ where
\[ | \mathcal{J}_{\text{discard}} | \le  \frac{4mT}{2^k} + 2^{qk}m\mathcal{C} \]
and every $j\in \mathcal{J}_{\text{top}} \setminus \mathcal{J}_{\text{discard}} $ is scheduled in the interval $[r_j,d_j].$
\end{theorem} 
%

In order to use the above Theorem, we need to find $r_j$ and $d_j$ satisfying the above four conditions. Before that, we first prove an easy bound on the length of a chain in $\mathcal{J}_{\text{top}} $. 

\begin{lemma}
\label{lem:chain}
The maximum chain length in $\mathcal{J}_{\text{top}} $ is at most $Ck\delta T$.
\end{lemma}
\begin{proof}
Each interval $I\in \mathcal{I}_\ell$ for $\ell \in [0,qk-1]$ has maximum chain length at most $\delta |I|$. Thus the maximum chain length in $ \mathcal{I}_\ell$ for $\ell \in [0,qk-1]$ is at most $\sum_{I\in\mathcal{I}_{\ell}} \delta |I|=\delta T$ and hence the maximum chain length in $\mathcal{J}_{\text{top}} = \bigcup_{\ell=0}^{(q-1)k-1}\mathcal{J}^{(s)}( \mathcal{I}_\ell)$ is at most $(q-1)k \delta T \le Ck\delta T$.
\end{proof}

We now find $r_j$ and $d_j$ satisfying the conditions of Theorem~\ref{thm:borrowed}. Recall that $S^{(s)}_j\subseteq I(j)$ and $I(j)$ is the smallest interval in the laminar family to satisfy this. Let $t_r(j)$ be the minimum index of the interval in $\mathcal{I}_{qk}$ which intersects $S^{(s)}_j$, that is, $t_r(j)= \min\{ p : I_p \cap S^{(s)}_j \neq \phi\}$. Similarly let $t_d(j)$ be the maximum index of the interval in $\mathcal{I}_{qk}$ which intersects $S^{(s)}_j$. Define
\[S'(j) =  I_{t_r(j)+1} \cup \dots \cup  I_{t_d(j)-1} .\] 
In other words, $S'(j)$ is obtained by chopping off from $S^{(s)}_j$ the first and the last intervals in $\mathcal{I}_{qk}$ intersecting $S^{(s)}_j$. We set $[r_j,d_j]:=S'_j$ for each $j\in \mathcal{J}_{\text{top}} $. \footnote{It is possible that $S'_j=\phi$ in which case we can take $r_j=d_j=m_j$. All these jobs will be discarded in Theorem~\ref{thm:borrowed}.}

Because $S^{(s)}_j\subseteq ext([r_j,d_j])$, conditions 1 and 4 in Theorem~\ref{thm:borrowed} follow straightaway. It only remains to prove conditions 2 and 3. We start with a useful lemma first. 

\begin{lemma}
\label{lem:meta}
Given any jobs $i$ and $j$ such that $j\prec i$. Then, in any round $s$ of the algorithm, $i$ cannot be assigned to a sub-interval of $I^{(s)}(j)_{\text{left}}$, and $j$ cannot be assigned to a sub-interval of $I^{(s)}(i)_{\text{right}}$.
\end{lemma}
\begin{proof}
Suppose to the contrary that in some round $s$, $i$ is assigned to a sub-interval of $I^{(s)}(j)_{\text{left}}$. Let $\ell$ be such that $I^{(s)}(j)\in\mathcal{I}_{\ell}$. 

Observe that $j$ cannot have any fractional support in $I^{(s)}(j)_{\text{right}}$, as then we would have a non-zero fraction of $j$ scheduled after $i$ has been fully scheduled, which contradicts the feasibility of the LP. Thus it must be the case that $\mathcal{F}^{(s)}_j\subseteq I^{(s)}(j)_{\text{left}}$.

Let $s'$ denote the last round of the algorithm when $j$ had a non-zero fractional support in $I^{(s)}(j)_{\text{right}}$. $s'$ is well defined because of the fact that $j$ is assigned to $I^{(s)}(j)$. The conditioning done after $s'$ which made $\mathcal{F}^{(s'-1)}_j \subseteq I^{(s)}(j)_{\text{left}} $ must have happened on an interval at a level below $\ell$ because $j$ remains assigned to level $\ell$. This means that if in round $s'$, $i$ was assigned to a level $\ell'\le\ell$, then $i$ also stays assigned to $\ell'$ from then on, and thus cannot get assigned to a sub-interval of $I^{(s)}(j)_{\text{left}}$ in round $s$. So it must be that in round $s'$, $i$ was assigned to a level $\ell'>\ell$.

But this implies that in round $s'$, $j$ had a non-zero fractional support in $I^{(s)}(j)_{\text{right}}$ while $\mathcal{F}^{(s')}_i \subseteq I^{(s)}(j)_{\text{left}}$. This contradicts the feasibility of the LP solution $y^{(s')}$. 

The other part of the lemma that $j$ cannot be assigned to a sub-interval of $I^{(s)}(i)_{\text{right}}$ follows similarly.
\end{proof}

We can now show that conditions 2 and 3 of Theorem~\ref{thm:borrowed} are satisfied. 

\begin{lemma}
If $i,j \in \mathcal{J}_{\text{top}} $ such that $j\prec i$, then $r_j\le r_i$ and $d_j\le d_i$. Thus, condition 2 is satisfied.
\end{lemma}
\begin{proof}
Suppose to the contrary that $r_j>r_i$. If $i$ has any fractional support to the left of $r_i$, which also means to the left of $S^{(s)}_j$, then that is a clear contradiction because then we would have a non-zero fraction of $i$ and no amount of $j$ scheduled before $r_i$. So assume that this is not the case. 

Because $S^{(s)}_i$ extends to the left of $r_i$, and hence to the left of the fractional support of $i$, it must be that $F^{(s)}_i \subseteq I(i)_{\text{right}}$ and $r_i=m_i+1$. In that case, $S^{(s)}_j \subseteq I(i)_{\text{right}}$ and thus $j$ is assigned to a sub-interval of $I(i)_{\text{right}}$, contradicting Lemma~\ref{lem:meta}. 

The proof for $d_j\le d_i$ follows similarly. Assume otherwise that $d_j>d_i$. If $j$ has any fractional support to the right of $d_j$, which also means to the right of $S^{(s)}_i$, that is a clear contradiction. So the only possibility is that $d_j=m_j$ and that $i$ is assigned to a sub-interval of $I(j)_{left}$, contradicting Lemma~\ref{lem:meta}.
\end{proof}

\begin{lemma}
For all $i\in \mathcal{J}'$ and $j\in \mathcal{J}_{\text{top}}$, if $j\prec i$ or $i\prec j$, then $\tilde{\sigma}(i) \not\in [r_j,d_j]$. Thus, condition 3 is satisfied.
\end{lemma}
\begin{proof}
Suppose $j\prec i$ and $\tilde{\sigma}(i) \in I_p\in \mathcal{I}_{qk}$. We argue that in this case $I_p \cap [r_j,d_j] =\phi$ and thus $\tilde{\sigma}(i) \not\in [r_j,d_j]$. The argument for the case $i\prec j$ follows similarly. 

Notice that because $\tilde{\sigma}(i) \in I_p\in \mathcal{I}_{qk}$, $i$ must have been assigned to a sub-interval of $I_p$ when we recursed on $I_p$. By Lemma~\ref{lem:meta}, it must be that $I_p \in I(j)_{\text{right}}$. Also, $j$ cannot have any fractional support to the right of $I_p$. These two facts imply that the right boundary of $S^{(s)}_j$, which is either the same as the right boundary of $\mathcal{F}^{(s)}_j$ or is at $(m_j+1)$, cannot be to the right of the right boundary of $I_p$. Thus $d_j$ is to the left of the left boundary of $I_p$. Hence $[r_j,d_j] \cap I_p=\phi$. 
\end{proof}

\begin{proof}(of Lemma~\ref{lem:discardrec1})
Because all conditions of Theorem~\ref{thm:borrowed} are satisfied, we get using Lemma~\ref{lem:chain}
\begin{eqnarray*}
 |\mathcal{J}_{\text{discard}}| &\le& \frac{4mT}{2^k} + 2^{qk}mCk\delta T \\
 & \le & \frac{4mT}{2^k} + 2^{Ck}mCk\delta T \\
 & = & \frac{\epsilon T}{8\log n}+ \frac{\epsilon T}{8\log n} =  \frac{\epsilon T}{4\log n}.
\end{eqnarray*}
\end{proof}

\subsubsection{Jobs discarded due to recursions of type $2$.}
Let $\epsilon'=\epsilon/4m$. Recall that in a recursion of type 2, we delete all the jobs assigned to levels $0$ to $(C-(1/\epsilon')^2)k-1$ and retain only the later $(1/\epsilon')^2$ batches. We show below that in such a case, at least a $(1-\epsilon')$ fraction of the jobs in the top $C$ batches are in the last $(1/\epsilon')^2$ batches and thus, by deleting the jobs in the first $C-(1/\epsilon')^2$ batches we only delete an $\epsilon'$ fraction of the jobs. 

\begin{lemma}
If case $(b)$ occurs in step 1 of the algorithm, then 
\begin{equation}
\label{eqn:discarded}
 \sum_{i=0}^{C-(1/\epsilon')^2-1}|\mathcal{J}^{(s)}(\mathcal{B}_i)|  \le \epsilon'  \sum_{i=C-(1/\epsilon')^2}^{C-1}|\mathcal{J}^{(s)}(\mathcal{B}_i)|  .
\end{equation} 
\end{lemma}
\begin{proof}
Let $S=\sum_{i=0}^{C-(1/\epsilon')^2-1}|\mathcal{J}^{(s)}(\mathcal{B}_i)| $, the left hand side of \eqref{eqn:discarded}. Case $(b)$ occurs in step 1 of the algorithm if none of the batches $\mathcal{B}_p$ for $p\in [ C-(1/\epsilon')^2 , C-1]$ are good. But then for $p\in [ C-(1/\epsilon')^2, C-1]$, we must have
\[ |\mathcal{J}^{(s)}(\mathcal{B}_{p})| > \epsilon' \sum_{i=0}^{p-1} |\mathcal{J}^{(s)}(\mathcal{B}_{i})| \ge \epsilon'  \sum_{i=0}^{C-(1/\epsilon')^2-1} |\mathcal{J}^{(s)}(\mathcal{B}_{i})| = \epsilon' S .\]
This implies \eqref{eqn:discarded} as
\[   \sum_{i=C-(1/\epsilon')^2}^{C-1}|\mathcal{J}^{(s)}(\mathcal{B}_i)| \ge \sum_{i=C-(1/\epsilon')^2}^{C-1} \epsilon' S = \left(\frac{1}{\epsilon'}\right)^2 \epsilon' S=\frac{S}{\epsilon'} . \]
\end{proof}

This implies that when we discard the top $C-(1/\epsilon')^2$ batches, we are only discarding at most an $\epsilon'$ fraction of the jobs in the next $(1/\epsilon')^2$ batches. We can imagine this as putting a charge of $\epsilon'$ on every job in the last $(1/\epsilon')^2$ batches. Thus the total charge on all the jobs at the end of the algorithm is an upper bound on the number of jobs discarded in recursions of type 2 during the algorithm.

\begin{lemma}
\label{lem:charge}
For every job $j$, we put a charge on $j$ at most once.
\end{lemma}
\begin{proof}
Fix a job $j$ and suppose we put a charge on $j$ at least once. When we put a charge on $j$ for the first time, then in some recursion of type $2$ it must have been assigned to the lowest $(1/\epsilon')^2$ batches among the top $C$ batches. The algorithm will then recurse on every interval at level $(C-(1/\epsilon')^2)k$ and thus job $j$ is now in the top $(1/\epsilon')^2$ batches in one of the recursive calls. 

Let $I\in \mathcal{I}_{(C-(1/\epsilon')^2)k}$ be such that $j$ is assigned to a sub-interval of $I$. When we recursively call the algorithm on $I$, the first $(1/\epsilon')^2$ batches already satisfy the property that any interval $I'$ in them has maximum chain length at most $\delta |I'|$. Thus in step $1$ of the algorithm, we will not condition on any interval in the top $(1/\epsilon')^2$ batches. This implies that job $j$ always stays assigned to the top $(1/\epsilon')^2$ batches; this is because the assignment of a job to an interval can only change when we condition on an interval at the same level or at a level above that of the job. 

Now suppose we put a charge on $j$ again. Then we must have once again done a recursion of type $2$ within the recursive call to $I$. But $j$ is assigned to the topmost $(1/\epsilon')^2$ batches in this instance and in a recursion of type $2$, we delete the topmost $C-(1/\epsilon')^2 > (1/\epsilon')^2$ batches and put a charge on only the later $(1/\epsilon')^2$ batches, which leads to a contradiction.
\end{proof}

\begin{lemma}
Number of jobs discarded in recursions of type 2 throughout the algorithm is at most $\epsilon T/2$.
\end{lemma}
\begin{proof}
Because the number of jobs discarded in recursions of type 2 throughout the algorithm is at most the total charge on all the jobs and by Lemma~\ref{lem:charge}, each job is charged at most once, we get that the number of jobs discarded in recursions of type 2 is at most
\[ \epsilon'n=\epsilon n/4m \le \epsilon T/4 \le \epsilon T/2.\]
\end{proof}

\section{Conclusion and Open Problems}
In this paper, we show that for constant $m$ and $\epsilon$, $\polylog(n)$ rounds of Sherali-Adams hierarchy reduce the integrality gap of the natural LP to $(1+\epsilon)$. A fascinating open problem is whether just $f(\epsilon,m)$ rounds of Sherali-Adams hierarchy, or the stronger Lasserre hierarchy, can also achieve the same approximation, for some function $f$ independent of the number of jobs. 

Another interesting open problem is to resolve the computational complexity of the problem. For $m=3$ and unit jobs, it is not even known if the problem is NP-hard.

Finally, another open problem is to get a $(1+\epsilon)$-approximation algorithm for arbitrary sized jobs and constant $m$. In this case, nothing better than a $(2-\frac{1}{m})$ approximation algorithm is known \cite{G66} to the best of our knowledge and hierarchies might help to close this gap.

\section{Acknowledgements}
We would to like to thank Seeun William Umboh, Martin B\"{o}hm and Nikhil Bansal for helpful discussions throughout this work.

\bibliographystyle{alpha}
{\small \bibliography{refr} }

\section{Appendix}
\label{sec:appendix}

In this section, we give the proof of Theorem~\ref{thm:borrowed}. This was already proved in \cite{R16} and we reproduce parts of their proof here. \\

\noindent{\bf Theorem~\ref{thm:borrowed} (\cite{R16}, restated).} {\em
Suppose we are given a feasible schedule $\tilde{\sigma}$ of the jobs in $\mathcal{J}'\subseteq \mathcal{J}_{\text{bottom}}$ and let the maximum chain length among jobs in $\mathcal{J}_{\text{top}}$ be $\mathcal{C}$. Suppose we are also given for each $j\in \mathcal{J}_{\text{top}}$ an interval $[r_j,d_j]$ such that: 
\begin{enumerate}
\item each $r_j,d_j$ coincides with the boundary of the intervals in $\mathcal{I}_{qk}$.
\item if $j\prec i$ for some $i\in \mathcal{J}_{\text{top}}$, then $r_j\le r_i$ and $d_j\le d_i$.
\item if $j\prec i$ or $i\prec j$ for some $i\in \mathcal{J}'$, then $\tilde{\sigma}(i)\not\in[r_j,d_j]$.
\item $\mathcal{F}^{(s)}_j\subseteq ext([r_j,d_j])$ and $m_j$ lies in the interior of $ext([r_j,d_j])$.
\end{enumerate}
Then, we can extend $\tilde{\sigma}$ to a feasible schedule $\sigma$ of $(\mathcal{J}_{\text{top}} \setminus \mathcal{J}_{\text{discard}}  ) \cup\mathcal{J}'$ where
\[ | \mathcal{J}_{\text{discard}} | \le  \frac{4mT}{2^k} + 2^{qk}m\mathcal{C} \]
and every $j\in \mathcal{J}_{\text{top}} \setminus \mathcal{J}_{\text{discard}} $ is scheduled in the interval $[r_j,d_j].$
} 

\begin{proof}
For $t\in [T]$, let $cap(t)=m-|\tilde{\sigma}^{-1}(t)|$ denote the number of empty slots at time $t$ in the schedule $\tilde{\sigma}$. Make a bipartite graph $G^+=(\mathcal{J}_{\text{top}}, U , E^{+})$ where one side of the bipartition consists of nodes representing jobs in $\mathcal{J}_{\text{top}}$, the other side $U$ consists of one node for each time $t$ with capacity $cap(t)$. The edge set $E^+$ is made as follows: for every $j\in \mathcal{J}_{\text{top}}$, we make an edge between the node for $j$ and the node for each $t\in ext([r_j,d_j])$.

Because of the first part of condition $4$, $\mathcal{F}^{(s)}(j)\subseteq ext([r_j,d_j])$, we know that there exists a perfect fractional matching of $\mathcal{J}_{\text{top}}$ in $G^+$, and thus there also exists a perfect integral matching in $G^+$. But this does not immediately give us a feasible schedule of $\mathcal{J}_{\text{top}}\cup  \mathcal{J}'$ as the schedule constructed in this way might not satisfy the precedence constraints.

To rectify this, we construct a somewhat smaller bipartite graph $G=(\mathcal{J}_{\text{top}}, U , E)$, where the edge set $E$ is made as follows: for all $j\in \mathcal{J}_{\text{top}}$, we make an edge between the node for $j$ and the nodes for each $t\in [r_j,d_j]$. Thus $G$ has a smaller neighbourhood for every $j\in \mathcal{J}_{\text{top}}$ than $G^+$. Notice that because of condition $3$, if we can find a perfect integral matching in $G$, then at least all the precedence constraints between the jobs in $\mathcal{J}_{\text{top}}$ and jobs in $\mathcal{J}'$ will be satisfied, even though they might not be satisfied among the jobs in $\mathcal{J}_{\text{top}}$. We will first show how to find a matching in $G$ of all but $4mT/2^k$ vertices in $\mathcal{J}_{\text{top}}$. Later we will take care of the precedence constraints among jobs in $\mathcal{J}_{\text{top}}$.

Let us see now how to bound the size of the maximum matching in $G$. By Hall's Theorem\cite{Sch}, the number of vertices of $\mathcal{J}_{\text{top}}$ which are left unmatched in a maximum matching in $G$ is exactly
\[  \max_{J \subseteq \mathcal{J}_{\text{top}}} \{ |J|-|N(J)|  \} \]
where $N(J)$ is the neighbour set of $J$ in $G$ and $|N(J)|$ denotes the sum of capacities of nodes in $N(J)$. Let $N^+(J)$ denote the neighbour set of $J$ in $G^+$. Again by Hall's Theorem, for any $J \subseteq \mathcal{J}_{\text{top}}$
\[  |J|-|N(J)|  \le  |N^+(J)|-|N(J)| . \]

Let the intervals at level $\mathcal{I}_{qk}$ be $I_1,\dots,I_{2^{qk}}$. First consider the case when for each $j\in \mathcal{J}_{\text{top}}$, $[r_j,d_j]\neq\phi$. Then for each job $j$, $N(\{j\})$ and $N^+(\{j\})$ is an interval of the form $I_a \cup I_{a+1} \cup\dots \cup I_b$, and hence for any $J \subseteq \mathcal{J}_{\text{top}}$, $N(J)$ and $N^+(J)$ are sets of disjoint intervals, where each interval is of the form $I_a \cup I_{a+1} \cup\dots \cup I_b$.  Fix a $J$ and notice $N^+(J) \subseteq \cup_A ext(A)$, where the summation is over all connected intervals $A$ of $N(J)$. Hence we can bound $|N^+(J)|-|N(J)| $ by the number of connected intervals of $N(J)$ times twice the maximum capacity of any interval in $\mathcal{I}_{qk}$, which is at most $2mT/2^{qk}$. 

Let the number of connected intervals of $N(J)$ be $c$. Then by the discussion in the previous paragraph, we can find a matching in $G$ which leaves at most $2cmT/2^{qk}$ jobs in $\mathcal{J}_{\text{top}}$ unmatched. We will now show that $c\le 2^{qk-k}$ which will give that at most $2mT/2^k$ jobs in $\mathcal{J}_{\text{top}}$ are left unmatched.

By condition 4, each $j\in \mathcal{J}_{\text{top}}$ has either $m_j\in [r_j,d_j]$ or $m_j+1\in[r_j,d_j]$ and thus, $N(\{j\})$ is a connected interval containing at least one of $m_j$ and $m_j+1$. Thus for an interval $I$ in the top levels, $N(\mathcal{J}^{(s)}(I) \cap J)$ is a single connected interval. This implies that the number of connected intervals of $N(J)$ is at most the number of intervals in the top levels which is at most
\[ \sum_{\ell=0}^{(q-1)k-1} 2^{\ell} \le 2^{(q-1)k} . \]

It remains to take care of the assumption that for each $j\in \mathcal{J}_{\text{top}}$, $[r_j,d_j]\neq\phi$. This won't happen only if for a job $j$, $N^+(\{j\})$, which consists of $ext([r_j,d_j])$, contains just one interval on either side around $m_j$. This is because of the assumption that $m_j$ lies in the interior of $ext([r_j,d_j])$. In this case, the fractional support of $j$ is also fully contained in the two intervals around $m_j$. Thus the number of such jobs can be at most the sum of the capacity of the intervals around the middle point of every interval in the top levels and is thus at most
\[ 2\frac{mT}{2^{qk}}  \sum_{\ell=0}^{(q-1)k-1} 2^{\ell} \le \frac{2mT}{2^k}. \]
We can discard all such jobs in the beginning and then in total we will have discarded at most $ \frac{4mT}{2^k}$ jobs from $\mathcal{J}_{\text{top}}$.

To recap, till now we have discarded $ \frac{4mT}{2^k}$ jobs from $\mathcal{J}_{\text{top}}$ and have made a schedule which might violate precedence constraints among the top jobs, but satisfies everything else required to prove Theorem~\ref{thm:borrowed}. To take care of the precedence constraints among the top jobs, we can now directly use Theorem 10 from \cite{R16} which, assuming conditions 1, 2 and 3, gives us a feasible partial schedule by discarding an extra $2^{qk}m\mathcal{C}$ jobs from $\mathcal{J}_{\text{top}}$. Moreover every job $j\in \mathcal{J}_{\text{top}}$ which is scheduled is scheduled within the interval $[r_j,d_j]$. This finishes the proof of Theorem~\ref{thm:borrowed}.
\end{proof}

\end{document}